%% file: brank.tex
\documentclass[11pt,a4paper]{article}

\usepackage{amssymb}
\usepackage{epsfig}
\usepackage{setspace}
\usepackage{amsmath,amssymb,amsthm}
\usepackage{mathrsfs}
\usepackage{color}
\usepackage{colortbl}
\usepackage{array}
\usepackage{braket}
\usepackage{authblk}
\usepackage{enumerate}
\usepackage{float}

\def\CC{{\bf C}}
\def\RR{{\bf R}}
\def\NN{{\bf N}}
\def\ZZ{{\bf Z}}
\def\QQ{{\bf Q}}

\def\numP{{\bf \#P}}

\def\cA{{\cal A}}

\def\cC{\mathcal{C}}
\def\cX{\mathcal{X}}
\def\cZ{\mathcal{Z}}

\def\cY{\mathcal{Y}}
\def\cI{\mathcal{I}}
\def\cJ{\mathcal{J}}
\def\cH{\mathcal{H}}
\def\cP{\mathcal{P}}

\newcommand{\Clow}{\mathop{\overline{\mathcal{C}}} }

\def\sF{\mathcal{F}}
\def\sG{\mathcal{G}}

\def\ff{\mbox{\boldmath $f$}}
\def\gvec{\mbox{\boldmath $g$}}
\def\vv{\mbox{\boldmath $v$}}
\def\uu{\mbox{\boldmath $u$}}
\def\xx{x}
\def\Mat{\mathrm{Mat}}
\def\Sym{\mathrm{Sym}}
\def\Psd{\mathrm{Psd}}

\numberwithin{equation}{section}

\newcommand{\dc}{\mathop{\rm dc} }
\newcommand{\brank}{\mathop{\text{{\rm b-rank}}}}

\newcommand{\perm}{\mathop{\rm perm} }

\newcommand{\rank}{\mathop{\rm rank} }

\newcommand{\conv}{\mathop{\rm conv} }

\newcommand{\tr}{\mathop{\rm tr} }
\newcommand{\prob}{\mathop{\rm prb} }

\newcommand{\sign}{\mathop{\rm sign}}

\newcommand{\minrank}{\mathop{\underline{\text{{\rm rank}}}}}

\theoremstyle{plain}
\newtheorem{thm}{Theorem}[section]
\newtheorem{lem}[thm]{Lemma}
\newtheorem{prop}[thm]{Proposition}
\newtheorem{cor}[thm]{Corollary}

\newtheorem{cnj}[thm]{Conjecture}
\newtheorem{fct}[thm]{Fact}

\theoremstyle{definition}
\newtheorem{df}[thm]{Definition}

\title{Bi-polynomial rank and determinantal complexity}
\author{Akihiro Yabe}
\affil{Graduate School of Information Science and Technology\\ 
The University of Tokyo\\
akihiro\_yabe@mist.i.u-tokyo.ac.jp
}

\date{March 31, 2015}

\begin{document}
\maketitle
\begin{abstract}
The permanent vs. determinant problem is one of the most
important problems in theoretical computer science,
and is the main target of geometric complexity theory proposed by 
Mulmuley and Sohoni.
The current best lower bound for the
determinantal complexity of the $d$ by $d$ permanent polynomial
is $d^2/2$, due to Mignon and Ressayre in 2004.
Inspired by their proof method,
we introduce a natural rank concept of polynomials, called 
the bi-polynomial rank.
The bi-polynomial rank is related to 
width of an arithmetic branching program.
%The bi-polynomial rank of a homogeneous polynomial $p$ of even degree $2k$
%is defined as the minimum $n$ such that $p$ can be
%written as a summation of $n$ products of polynomials of degree $k$.
We prove that the bi-polynomial rank gives a lower bound of the determinantal complexity.
As a consequence, the above Mignon and Ressayre bound is improved 
to $(d-1)^2 + 1$ over the field of reals.
We show that the computation of the bi-polynomial rank is formulated as a rank minimization problem.
%Applying the concave minimization technique,
%we reduce the problem of lower-bounding determinantal complexity
%to that of proving the positive semidefiniteness of matrices,
%and this is a new approach for the permanent vs. determinant problem.
We propose a computational approach for
giving a lower bound of this rank minimization,
via techniques of the concave minimization.
This also yields a new strategy to attack the permanent vs. determinant problem.
\end{abstract}
\input{brank_intro}

\input{brank_basic}

\input{brank_bound}
\input{brank_conjecture}
\input{brank_proof1}
\input{brank_proof2}

\input{brank_acknowledgement}

\bibliographystyle{fplain}
\bibliography{../../../bibtex/all}

%\appendix
%\input{brank_appendix}
\end{document}

%% file: brank_intro.tex
\section{Introduction}\label{sec_intro}
The determinant $\det(A)$ and the permanent $\perm(A)$ of
a square matrix $A = (a_{i,j})$ of size $d$ are defined by
\begin{align*}
	{\det} (A) &:= \sum_{\sigma \in \mathfrak{S}_d} \sign(\sigma) \prod_{i=1}^d a_{i, \sigma(i)},\\
	{\perm} (A) &:= \sum_{\sigma \in \mathfrak{S}_d} \prod_{i=1}^d a_{i, \sigma(i)},
\end{align*}
where $\mathfrak{S}_d$ is the set of permutations on $\{1,2,\dots,d\}$.
Determinant is a representative function which admits efficient computation only with arithmetic operations.
On the other hand, such an efficient computation for permanent is not known.
Valiant~\cite{Valiant1979Tco} proved that the computation of permanent of $0$-$1$ matrices is $\numP$-complete.
Therefore, in contrast to determinant, 
it is conjectured that permanent cannot be computed in polynomial time.

The determinantal complexity is a measure for the difficulty of evaluation of polynomials.
%Let $K$ be a field, 
%$M_n(K)$ be the set of square matrices on $K$ of size $n$, 
%and $\det$ be determinant. 
Let $K$ be a field, 
and let $K[x] = K[x_1,x_2,\dots,x_D]$ denote
the set of polynomials of variables $x_1,x_2,\dots,x_D$ with coefficients in $K$.
By an {\em affine polynomial matrix}, 
we mean a matrix each of whose entries is an affine polynomial
(a linear polynomial including a constant term).
\begin{df}[see~\cite{Mignon2004Aqb}]
The {\em determinantal complexity} $\dc(p)$ of $p \in K[x]$ is defined as the minimum number $n$ such that there exists an affine polynomial matrix $Q \in (K[x])^{n \times n}$ satisfying
\[
p = \det(Q).
\]
\end{df}
It is known in~\cite{Valiant1979Cci, Valiant1983Fpc} that 
if a polynomial $p$ can be evaluated with number $m$ of arithmetics, then the determinantal complexity $\dc(p)$ is $O(m^{c \log m})$ for some constant $c$.
%, where it is assumed that the degree of $p$ is bounded by a polynomial on the number of variables appearing in $p$~\cite{Valiant1979Cci, Valiant1983Fpc}. 

Permanent is regarded as a polynomial of matrix entries.
Let $\perm_d$ denote the permanent polynomial for $D = d \times d$ variables of matrix entries.
If $\dc(\perm_d) =  d^{ \omega(\log d) } $ over $K$, then permanent cannot be computed by polynomial number of arithmetics on $K$.
The following is one of the main conjectures in algebraic complexity theory (see~\cite{Burgisser1997Alt}).
\begin{cnj}
Over a field $K$ of characteristic not equal to two, it holds that
\[
\dc({\perm}_d) =  d^{ \omega(\log d) } .
\]
\end{cnj}
This conjecture implies ${\bf VP}_K \neq {\bf VNP}_K$, an arithmetic counterpart of ${\bf P}$ vs. ${\bf NP}$ conjecture,
since permanent is in ${\bf VNP}_K$ (in fact ${\bf VNP}_K$-complete) if the characteristic of $K$ is not equal to two~\cite{Valiant1979Cci}.
%If this conjecture is true, then ${\bf VP}_K \neq {\bf VNP}_K$, which is the arithmetic counterpart of ${\bf P}$ vs. ${\bf NP}$ conjecture.
%This holds since permanent is in ${\bf VNP}_K$ (in fact ${\bf VNP}_K$-complete) if the characteristic of $K$ is not equal to two~\cite{Valiant1979Cci}. 
%
%Geometric complexity theory (GCT)
%proposed by Mulmuley and Sohoni~\cite{Mulmuley2001GCT1}
%; see~\cite[Chapter 13]{Landsberg2012Tga}.

%The current best lower bound for $\dc(\perm_d)$ is quadratic, which was proved by Mignon and Ressayre~\cite{Mignon2004Aqb}.
The current best lower bound for $\dc(\perm_d)$, due to Mignon and Ressayre~\cite{Mignon2004Aqb},
is quadratic.
%, and far from the conjecture. 
\begin{thm}[Mignon and Ressayre~\cite{Mignon2004Aqb}] \label{thm_mignon}
Over a field $K$ of characteristic zero, it holds that
\[
\dc({\perm}_d) \geq \frac{d^2}{2}.
\]
\end{thm} 
Improving this bound is one of the most prominent issues in the literature.
%This theorem was proved by comparing the rank of Hessians of determinant and permanent.
%There are two extensions of this theorem. 
Cai, Chen and Li~\cite{Cai2010Qlb} proved that $\dc(\perm_d) \geq (d-2)(d-3)/2$ over any field $K$ of characteristic not equal to two. 
Mulmuley and Sohoni~\cite{Mulmuley2001GCT1}
proposed a magnificent program,
called {\em geometric complexity theory} (GCT),
to obtain super-polynomial lower bounds
by utilizing deep techniques of 
algebraic geometry and representation theory (also, see~\cite[Chapter 13]{Landsberg2012Tga}).
In the context of GCT,
Landsberg, Manivel and Ressayre~\cite{Landsberg2013Hwd} proved that the same lower bound $d^2/2$ holds for 
the orbit closure version $\overline{\dc}$ of the determinantal complexity.
%according to the philosophy of
%geometric complexity theory (GCT) proposed by Mulmuley and Sohoni~\cite{Mulmuley2001GCT1} (also, see~\cite[Chapter 13]{Landsberg2012Tga}).

%
%
\paragraph{Our contribution.}
We introduce the bi-polynomial rank
of a homogeneous polynomial of even degree,
and prove that the determinantal complexity is
bounded below by the bi-polynomial rank.
Our technique may be viewed as 
a higher order generalization of the Hessian rank comparison
proof of the above $d^2/2$ bound (Theorem~\ref{thm_mignon})
by Mignon and Ressayre.
Let $K[x]^{(k)} \subseteq K[x]$ denote the set of homogeneous polynomials %with 
of degree $k$.
\begin{df}\label{df_brank}
The {\em bi-polynomial rank} $\brank(p)$ of $p\in K[x]^{(2k)}$ is defined as the minimum number $n$ such that 
there exist $2n$ polynomials $f_1,f_2,\dots,f_n$, $g_1,g_2,\dots,g_n \in K[x]^{(k)}$ satisfying 
\[
p = \sum_{i=1}^n f_ig_i.
\]
\end{df}
%Our main theorem gives an explicit relation between the
%bi-polynomial rank and the determinantal complexity.
%Let $p \in K[x]$. 
For $p \in K[x]$ and $x_0 \in K^D$, 
we define a polynomial $p_{x_0}$ by $p_{x_0}(x) := p(x+x_0)$.
We denote by $p_{x_0}^{(k)}$ the degree-$k$ homogeneous part of $p_{x_0}$.
%For $x_0 \in K^D$, let $p_{x_0}$ be the polynomial defined by $p_{x_0}(x) := p(x+x_0)$. 
The set of points $z\in K^D$ with $p(z)=0 $ is denoted by ${\rm Zeros}(p)$.
% ${\rm Zeros}(p) \subseteq K^D$ is defined by ${\rm Zeros}(p) := \{ z \in K^{D} \mid p(z)=0 \}$. 
Our main result is the following.
%The main property of the bi-polynomial rank is the following.
\begin{thm}\label{thm_main_dc}
For a polynomial $p \in K[x]$, $k \in [1,D]$, and $x_0 \in {\rm Zeros}(p)$, it holds that
\begin{align*}
\dc(p) \geq \frac{1 }{2^{2k-2}}\brank(p_{x_0}^{(2k)}) - 2(k-1) D^{k-1}. 
%\left(k \leq D,\  x_0 \in {\rm Zeros}(p) \right). 
\end{align*}
\end{thm}
We will see in Section~\ref{sec_dim_poly} that every generic polynomial $p \in K[x]^{(2k)}$ has the bi-polynomial rank at least $k! D^k / 2(2k)!$.
This means that the bi-polynomial rank has a potential
to give $\Omega((d^2)^k)$ lower bound to $\dc(\perm_d)$
for every $k$.
%Therefore the second term on the right hand side is negligible for generic polynomials with $D \gg k$.
%A direct corollary of Theorem~\ref{thm_main_dc} for the permanent vs. determinant problem
%is the following.
A direct implication to the permanent vs. determinant 
problem is the following.
\begin{cor}\label{cor_main}
	Let $k \geq 1$ be an arbitrary integer.
	If there exists a sequence of matrices
	$X_d \in {\rm Zeros}(\perm_d)$ for $d=1,2,\dots$ 
	%$\{X_d \mid X_d \in {\rm Zeros}(\perm_d) \}_{d=1}^{\infty}$ 
	such that $\brank(\perm_{d,X_d}^{(2k)}) = \Omega(d^{2k})$, 
	then $\dc(\perm_{d}) = \Omega(d^{2k})$.
\end{cor}

In the case $k=1$, our approach sharpens
the Hessian approach by Mignon and Ressayre.
%We show that Theorem~\ref{thm_main_dc} implies Theorem~\ref{thm_mignon} in Section~\ref{sec_demonstration}.
We will see in Section~\ref{sec_demonstration}
that Theorem~\ref{thm_main_dc} directly implies Theorem~\ref{thm_mignon}.
Furthermore, over the field $\RR$, our approach improves the quadratic bound as follows.
\begin{thm}\label{thm_d2bound}
Over the field $\RR$, it holds that
\[
\dc({\perm}_d) \geq (d-1)^2 +1.
\]
\end{thm}

\paragraph{Bounding b-rank via concave minimization.}
In the case $k \geq 2$, the direct calculation of the bi-polynomial rank is still difficult.
We propose the following computational procedure 
to bound the bi-polynomial rank over $\RR$.
Let $\Sym_n$ and $\Psd_n$ denote the sets of
real symmetric and positive semidefinite matrices of size $n$, respectively.
%For a symmetric matrix $X$, we denote by $\mu_l(X)$ the sum of smallest $l$ eigenvalues of $X$.
%
Suppose $p \in \RR[x]^{(2k)}$. 
We will show that
$\brank(p)$ is 
at least the half of the minimum rank of a matrix of size $n=2s_k$ 
in $\cX_p \cap \Psd_n$,
where $s_k := \binom{D+k-1}{k}$ and $\cX_p$ is an affine subspace in $\Sym_n$
explicitly represented by linear equations determined by the coefficients of $p$;
see the concrete definition in Section~\ref{section_CB}. 
Thus $\brank(p)$ is more than $r/2$
if the sum $\mu_{n-r}(X)$ of the smallest $n-r$ eigenvalues of $X$
is positive for all $X \in \cX_p \cap \Psd_n$.
It is known that the function $X \mapsto \mu_{n-r}(X)$ is concave on $\Sym_n$.
A well-known fact in concave function minimization theory~\cite{Pardalos1986Mfg}
tells us that if we know a polyhedral convex set $\cP \subseteq \Sym_n$
containing $\cX_p \cap \Psd_n$,
then the minimum of $\mu_{n-r}$ is attained by extreme points of $\cP$.
Thus, the positivity of $\mu_{n-r}$ for all these extreme points is 
a certificate of $\brank(p) \geq r/2$.
\begin{prop}\label{prop_concave_bound}
	Let $p \in  \RR[x]^{(2k)}$, $r \in \NN$, and $n = 2s_k$. 
	If there exists $\cY \subseteq \Sym_{n}$ satisfying
	the following property, then $\brank(p) > r/2$.
	
	(i) $\cX_p \cap \Psd_{n} \subseteq \conv(\cY)$.
	
	(ii) $\mu_{n-r}(Y) > 0$ for all $Y \in \cY$.
\end{prop}
It should be noted that 
this approach is essentially an outer approximation algorithm~\cite{Kelley1960tcp,Tuy1983Ooa} in the concave minimization.

\paragraph{Related work.}
%For $p \in K[x]$ and $x_0 \in {\rm Zeros}(p)$, 
The bi-polynomial rank $\brank(p_{x_0}^{(2k)})$ can be interpreted
as the minimum width of the $k$th layer of an arithmetic branching program (ABP) computing $p_{x_0}^{(2k)}$.
Since the determinant polynomial of a matrix of size $n$ has an ABP with width at most $n^2$~\cite{Mahajan1997Dca}, it directly follows that a simple but weaker bound $\dc(p) \geq \sqrt{\brank(p_{x_0}^{(2k)})}$.
We include the detailed discussion in Section~\ref{Subsec_ABP}.
Our bound shows the possibility to prove $\dc(\perm_d) = \Omega(d^4)$ by considering forth derivatives $\perm_{d,X_d}^{(4)}$ of $\perm_d$,
which seems significantly simpler than considering eighth derivatives.
%Observe, however, that
%this immediate bound is useless for proving $\dc(\perm_d) = \Omega(d^4)$ by considering forth derivatives of $\perm_d$ which looks significantly simpler than eighth derivatives.

Our proof method of Theorem~\ref{thm_main_dc} is first considering a normal form of an affine polynomial matrix $Q$,
and then constructing an ABP of $\det(Q)^{(2k)}$ with small width using an exhaustive construction of low-degree terms.
Such an exhaustive construction implicitly appears in the area of depth reduction of arithmetic circuits~\cite{Valiant1983Fpc}.

Nisan~\cite{Nisan91lowerbounds} 
considered the rank of a matrix defined by partial derivatives of non-commutative
determinant,
and proved an exponential lower bound of the size of
ABP of non-commutative determinant. 
This implies an exponential lower bound of the size of non-commutative formulas for determinant.
We consider the bi-polynomial rank, which is width of ABP, and formulate the bi-polynomial rank
as the minimum matrix-rank over an affine subset of matrices.
Therefore our approach may be viewed as a commutative analogue of Nisan's approach.

The difficulty of lower bound problems come from 
that of proving non-existence of certain objects.
The essential idea in GCT~\cite{Mulmuley2001GCT1} is to flip the non-existence
of embeddings into
the existence of representation-theoretical obstructions.
Our approach might yield 
a comparable optimization-theoretic flip strategy 
for the permanent vs. determinant problem: for proving $\dc(\perm_d) = \Omega(d^{2k})$,
\begin{quote}
	find $X_d \in {\rm Zeros}(\perm_d)$, a polyhedron $\cP_d$ containing 
	$\cX_p \cap \Psd_n$ for $p = \perm_{d,X_d}^{(2k)}$,
	and $r = O(d^{2k})$ such that $\mu_{n - r}(P) >0$
	holds for all extreme points $P$ of $\cP_d$.
\end{quote}
Though much still remains to be unsettled,
we hope that
our approach will bring a new inspiration and
trigger a new attack to this extremely difficult lower bound issue.

\paragraph{Organization.}
In Section~\ref{sec_basic_properties}, we prove basic properties of the bi-polynomial rank.
In Section~\ref{section_CB}, 
we introduce a formulation of the bi-polynomial rank as the minimum matrix-rank over an affine subspace of matrices.
%we show that for a fixed degree $k$, 
%the computation of the bi-polynomial rank $\brank(p)$ of $p \in K[x]^{(2k)}$
%can be reduced to a rank minimization problem of polynomial size.
In Section~\ref{sec_dim_poly}, we prove that generic polynomials $p \in K[x]^{(2k)}$
have the bi-polynomial rank at least $D^k / (2k)!$.
In Section~\ref{Subsec_ABP}, we discuss a relation between the bi-polynomial rank
and ABP.
In Section~\ref{sec_relation_b-rank_detcomp}, 
we consider lower bounds of $\dc(\perm_d)$ from the bi-polynomial rank for the case $k=1$.
%we prove lower bounds of $\dc(\perm_d)$ by using the bi-polynomial rank for the case $k=1$.
In Section~\ref{sec_demonstration}, 
we demonstrate that the bi-polynomial rank generalizes the Hessian rank,
and give an alternative and conceptually simpler proof of Theorem~\ref{thm_mignon}.
%we give an alternative proof of the quadratic bound given by Mignon and Ressayre for demonstrating the bi-polynomial rank as a generalization of the Hessian rank.
In Section~\ref{subsec_newbound}, we prove 
$\dc(\perm_d) \geq (d-1)^2 + 1$ over the real field (Theorem~\ref{thm_d2bound}).
%Theorem~\ref{thm_d2bound} which is the improved lower bound over the field $\RR$.
%
In Section~\ref{sec_cnj}, we prove Proposition~\ref{prop_concave_bound}, and propose an approach for the permanent vs. determinant problem based on the bi-polynomial rank, the rank minimization, and the concave minimization for $k \geq 2$.
In Section~\ref{sec_proof},
we prove Theorem~\ref{thm_main_dc}, the main result of this paper.

%% file: brank_basic.tex
\section{Basic properties of b-rank}\label{sec_basic_properties}
\subsection{Rank minimization for b-rank}\label{section_CB}
To consider the calculation of the bi-polynomial rank, we formulate the bi-polynomial rank as the minimum matrix-rank over an affine subspace of matrices.
This formulation is a basis for discussions in subsequent sections.

Let $\Mat_n(K)$ be the set of square matrices of size $n$ over the field $K$. 
For a nonnegative integer $k$,
let $\cI_k (=\cI_{k,D})$ denote the set of 
$D$-tuples $(i_1,i_2,\dots,i_D)$
of nonnegative integers such that 
the sum $i_1 + i_2 + \cdots + i_D$ is equal to $k$.
For $I = (i_1,i_2,\dots,i_D) \in \cI_k$,
let $x^{I}$ denote the monomial $x_1^{i_1}x_2^{i_2}\cdots x_D^{i_D}$.
We define $s_k(={s_{k,D}}) := |\cI_k| = \binom{D + k -1}{k}$,
and consider that $s_k$-dimensional vectors $\uu = (u_I)_{I \in \cI_k}$ and matrices $Q = (q_{I,J})_{I,J \in \cI_k}$ of size $s_k$ are indexed by elements of $\cI_k$.
Then their products are written as $(Q\uu)_I = \sum_{J \in \cI_k} q_{I,J} u_{J}$.
Let $\vv(x) := (x^I)_{I \in \cI_k}$ be the $s_k$-dimensional vector which consists of monomials.
\begin{thm}\label{thm_to_opt}
For $p \in K[x]^{(2k)}$, $\brank(p)$ is equal to the optimum value of the following problem:
\begin{align*}
\begin{array}{ll}
{\rm Minimize} \quad &\rank(Q) \\
{\rm subject \ to} \quad & p(x) = \vv(x)^{\top} Q \vv(x), \\
  & Q \in \Mat_{s_k}(K).
\end{array}
\end{align*}
\begin{proof}
Suppose that $Q_{{\rm opt}}$ attains the optimum value.
First we prove that $\brank(p) \leq \rank(Q_{{\rm opt}})$.
Let $r := \rank(Q_{{\rm opt}})$. We can represent $Q_{{\rm opt}}$ as a sum $Q_{{\rm opt}} = \sum_{i=1}^r \ff_i \gvec_i^{\top}$ of rank one matrices $\ff_i \gvec_i^{\top}$,
where $\ff_i$ and $ \gvec_i$ are $s_k$-dimensional vectors for $i=1,2,\dots,r$.
Then we have 
\[
p(x) = \vv(x)^{\top} Q_{{\rm opt}} \vv(x) = \sum_{i=1}^{r} (\ff_i^{\top} \vv(x))( \gvec_i^{\top} \vv(x)).
\]
Choosing $2r$ polynomials 
$f_1,f_2,\dots,f_r,g_1,g_2,\dots,g_r \in K[x]^{(k)}$ 
as $f_i (x) = \ff_i^{\top} \vv(x)$ and $g_i (x) = \gvec_i^{\top} \vv(x)$ for $i=1,2,\dots,r$, we have $\brank(p) \leq r = \rank(Q_{{\rm opt}})$.

Next, we show that $\brank(p)\geq \rank(Q_{{\rm opt}})$.
Set $n := \brank(p)$. 
From the definition, there exist $2n$ polynomials $f_1,f_2,\dots,f_n,g_1,g_2,\dots,g_n \in K[x]^{(k)}$
such that $p = \sum_{i=1}^n f_i g_i$.
Suppose that $f_i(x) = \sum_{I \in \cI_k} f_{i,I} x^I$, for $f_{i,I} \in K$, where $i=1,2,\dots,n$.
Then we can represent polynomials $f_i$ by inner product of $s_k$-dimensional vectors as $f_i(x) = \ff_i^{\top} \vv(x)$, where $\ff_i := (f_{i,I})_{I \in  \cI_k}$.
Also, we can represent $g_i$ as $g_i(x) = \gvec_i^{\top} \vv(x)$ where $\gvec_i := (g_{i,I})_{I \in  \cI_k}$.
Then we have 
\[
p(x) = \sum_{i=1}^n (\ff_i^{\top} \vv(x))(\gvec_i^{\top} \vv(x)) = \sum_{i=1}^n \vv(x)^{\top} (\ff_i \gvec_i^{\top}) \vv(x).
\]
Defining $Q_0 := \sum_{i=1}^n \ff_i \gvec_i^{\top}$, it follows that $\vv(x)^{\top} Q_0 \vv(x) = p(x)$.
Thus $Q_0$ satisfies the constraints, and we have $\brank(p) = n \geq \rank(Q_0) \geq \rank(Q_{{\rm opt}})$.
\end{proof}
\end{thm}
Observe that the feasible region of the above problem is an affine subspace of the set of matrices.
We give similar formulations over symmetric and positive semidefinite matrices over $\RR$.
%
%
%
%\begin{cor}\label{cor_to_opt}
%Let $p(x) = \sum_{H \in \bbinom{D}{2k}} a_H x^{H} \in K[x]^{(2k)}$. Then $\brank(p)$ is equal to the optimum value of the following problem:
%\begin{align*}
%\begin{array}{ll}
%{\rm Minimize} \quad &\rank(Q) \\
%{\rm subject \ to} \quad & \sum_{\substack{I,J \in \bbinom{D}{k} \\ I+J = H } } q_{I,J}
%= a_H \quad {\rm for}\ H \in \bbinom{D}{2k}, \\
%  & Q \in \Mat_{s_k}(K).
%\end{array}
%\end{align*}
%\begin{proof}
%This follows from
%\[
%\vv(x)^{\top} Q \vv(x) = \sum_{I,J \in \bbinom{D}{k}} Q_{I,J} x^I x^J
%= \sum_{H \in \bbinom{D}{2k}} x^{H} \sum_{\substack{I,J \in \bbinom{D}{k} \\ I+J = H } } Q_{I,J}.
%\]
%\end{proof}
%\end{cor}
%
\begin{cor}\label{cor_brank_sym_rankmin}
	For $p \in \RR[x]^{(2k)}$, $\brank(p)$ is at least the half of the optimum value of the following problem:
	\begin{align*}
	\begin{array}{ll}
	{\rm Minimize} \quad &\rank(Q) \\
	{\rm subject \ to} \quad & p(x) = \vv(x)^{\top} Q \vv(x), \\
	& Q \in \Sym_{s_k}.
	\end{array}
	\end{align*}
	\begin{proof}
		Consider the optimum solution $Q' \in \Mat_n(\RR)$ of the corresponding optimization problem in Theorem~\ref{thm_to_opt}.
		Then it holds that $\brank(p) = \rank(Q')$.
		Since $Q = (Q' + Q'^{\top})/2$ is a feasible solution of the above problem and $\rank(Q') \geq \rank(Q)/2$, the statement holds.
	\end{proof}
\end{cor}
\begin{cor}\label{cor_brank_psd_rankmin}
	For $p \in \RR[x]^{(2k)}$, $\brank(p)$ is at least the half of the optimum value of the following problem:
	\begin{align*}
	\begin{array}{ll}
	{\rm Minimize} \quad &\rank(Q_+) + \rank(Q_-) \\
	{\rm subject \ to} \quad & p(x) = \vv(x)^{\top} (Q_+ - Q_-) \vv(x), \\
	& Q_+, Q_- \in \Psd_{s_k}.
	\end{array}
	\end{align*}
	\begin{proof}
		Since any symmetric matrix $Q \in \Sym_n$ can be uniquely represented as the difference $Q=Q_+ - Q_-$ of the two positive semidefinite matrices $Q_+, Q_- \in \Psd_n$
		satisfying $\rank(Q) = \rank(Q_+) + \rank(Q_-)$,
		the statement follows from Corollary~\ref{cor_brank_sym_rankmin}.
	\end{proof}
\end{cor}
For $p \in \RR[x]^{(2k)}$, we define $\cX_p$
as the set of pairs $(Q_+, Q_-)$ of $s_k \times s_k$ matrices
satisfying linear equation $p(x)=\vv(x)^{\top} (Q_+ - Q_-) \vv(x)$.
By the embedding
\[
(Q_+, Q_-) \mapsto 
\left(
\begin{array}{cc}
Q_+ &O \\
O &Q_-
\end{array}
\right).
\]
we regard $\cX_p$ as an affine subspace of $\Sym_{2s_k}$.
Then $\cX_p \cap \Psd_{2s_k}$ is the feasible region of the optimization problem in Corollary~\ref{cor_brank_psd_rankmin}.
%,and thus the concrete definition of $\cX_p$ in Section~\ref{sec_intro} is given.

%In particular, the affine space $\cX_p$ (in Section~\ref{sec_intro})
%is given by the set of $2s_k \times 2s_k$ matrices
%\[
%\left(
%\begin{array}{cc}
%Q_+ &O \\
%O &Q_-
%\end{array}
%\right)
%\]
%satisfying linear equation $\vv(x)^{\top} (Q_+ - Q_-) \vv(x)$.
%Then $\cX_p \cap \Psd_{2s_k}$ is identified with the feasible region of the optimization %
%problem in Corollary~\ref{cor_brank_psd_rankmin}
%via the embedding
%\[
%(Q_+, Q_-) \mapsto 
%\left(
%\begin{array}{cc}
%Q_+ &O \\
%O &Q_-
%\end{array}
%\right).
%\]
%%we regard $\cX_p$ as a subspace of $\Sym_{2s_k}$.

\subsection{b-rank of generic polynomials}\label{sec_dim_poly}
The inequality in Theorem~\ref{thm_main_dc} 
is nontrivial
only if the bi-polynomial rank is larger than $2^{2k-1}(k-1)D^{k-1}$.
We are going to show that for $D \gg k$ and a generic polynomial, 
this condition holds.
%In this section, we consider the dimension of the set of polynomials having bi-polynomial rank at most $r$.
%Proposition~\ref{prop_basic_bounds} follows from this statement. 
Here we suppose that $K$ is an algebraically closed field of characteristic zero.
We use the terminologies in Section~\ref{section_CB}.
Observe that polynomials in $K[x]^{(2k)}$ are determined by $s_{2k} $ coefficients, and therefore we regard that a homogeneous polynomial $p \in K[x]^{(2k)}$ is a point in $K^{s_{2k}}$, under the correspondence $p(x) = \sum_{I \in \mathcal{I}^{(2k)}_D}a_I x^I \mapsto (a_I)_{I \in \mathcal{I}^{(2k)}_D} \in K^{s_{2k}}$.
Then the set of polynomials $p$ satisfying $\brank(p) \leq r$ are characterized in terms of algebraic geometry, as follows.
\begin{thm}\label{thm_dim_birank}
	Let $S := \{q \in K[x]^{(2k)} \mid \brank(q) \leq r  \} \subseteq K^{s_{2k}}$.
	Then the Zariski closure $\overline{S}$ is an irreducible variety having dimension at most $r(2s_k-r)$.
	\begin{proof}
		We identify $\Mat_{s_k}(K)$ with $K^{s_k^2}$.
		Let $Z_r := \{ X \in K^{s_k^2} \mid \rank(X) \leq r \}$.
		$Z_r$ is called the {\em determinantal variety}, and it is known that $Z_r$ is an irreducible variety of dimension $r(2s_k - r)$ (see, e.g.,~\cite{Harris1995Ag}).
		%We refer to the elements of $X = (X_{I,J})_{I,J \in \bbinom{D}{k}} \in K^{s_k^2} $ 
		%and $u = (u_H)_{H \in \bbinom{D}{2k}} \in K^{s_{2k}}$ by elements in $\bbinom{D}{k}$ and $\bbinom{D}{2k}$, respectively.
		We define $\pi:K^{s_k^2} \to K^{s_{2k}}$ by 
		\[
		(\pi(X))_H := \sum_{I,J:I+J = H} X_{I,J} \quad  (I,J \in \cI_k,\ H \in \cI_{2k}) . 
		\]
		%With a linear transformation of coordinates, t
		This $\pi$ is a linear projection.
		For $q \in K[x]^{(2k)}$,
		Theorem~\ref{thm_to_opt} shows that $\brank(q) \leq r$
		if and only if there exists $X \in K^{s_k^2}$ such that $\rank(X) \leq r$
		and $\pi(X) = q$.
		Thus it follows that $S = \pi(Z_r)$.
		As the Zariski closure of the linear projection of the irreducible variety $Z_r$,
		$\overline{S} = \overline{\pi(Z_r)}$ is irreducible and
		its dimension is at most $r(2s_k - r)$.
		%Since $I(Z_r)$ is a prime ideal,
		%from the closure theorem (see~\cite[Chapter 3]{Cox2007Ait})
		%$I(S)$ is also a prime ideal, 
		%and therefore $\overline{S}$ is irreducible.
		%Since $\pi$ is a projection, the dimension of $\overline{S}$ is at most $r(2s_k - r)$.
	\end{proof}
\end{thm}
From Theorem~\ref{thm_dim_birank}, we can obtain a lower bound
of the bi-polynomial rank for generic polynomials.
\begin{prop}\label{prop_basic_bounds}
	For a polynomial $p \in \CC[x]^{(2k)}$ with algebraically independent coefficients over $\QQ$, it holds that $\brank(p) \geq k! D^k /2(2k)!$. 
	\begin{proof}
		Let $r := \brank(p)$ and $S_r := \{q \in \CC[x]^{(2k)} \mid \brank(q) \leq r  \} \subseteq \CC^{s_{2k}}$. From Theorem~\ref{thm_dim_birank}, $\overline{S_r}$ is an irreducible variety of dimension at most $r(2s_k - r)$. 
		%Since the ideal of the determinantal variety on $\CC$ is generated by polynomials with coefficients in $\QQ$, from the proof of Theorem~\ref{thm_dim_birank} and the closure theorem, the ideal $I(S_r)$ is also generated by polynomials with coefficients in $\QQ$.
		Since $Z_r$ and hence $\overline{S}_r$ is defined over $\QQ$
		and $p$ has no algebraic relation over $\QQ$ in its coefficients,
		$p \in \overline{S_r}$ implies that $\overline{S_r}$ must be $\CC^{s_{2k}}$.
		Comparing the dimensions, it must holds that $r(2 s_k-r) \geq s_{2k}$.
		This implies $r^2 - 2s_k r + s_{2k} \leq 0$ and
		\begin{align*}
		r \geq s_k - \sqrt{s_k^2 - s_{2k}}
		\geq s_k 
		\left(1- \left(1 - \frac{s_{2k}}{2s_k^2} \right) \right)
		=\frac{s_{2k}}{2 s_k} \geq \frac{k!D^k}{2(2k)!}. 
		\end{align*}
		In the second inequality, we use the fact that $\sqrt{1-x} \leq 1 - x/2$ for $x \leq 1$.
	\end{proof}
\end{prop}
%Note that this is the best bounds with respect to the degree of $D$, because of the following proposition.
This lower bound is asymptotically tight if $k$ is a constant.
\begin{prop}
	For any polynomial $p \in K[x]^{(2k)}$, it holds that $\brank(p) \leq s_k \leq D^k$.
\end{prop} 
This proposition immediately follows from Lemma~\ref{lem_basic}.

\subsection{b-rank and arithmetic branching program}\label{Subsec_ABP}
We here discuss a relation between the bi-polynomial rank and an arithmetic branching program (ABP).
We show that the following weaker statement than Theorem~\ref{thm_main_dc}
easily follows from known facts on an ABP of determinant.
\begin{prop}\label{prop_sqrt}
	For a polynomial $p \in K[x]$, $k \in \NN$, and $x_0 \in K^D \setminus {\rm Zeros}(p)$, it holds that
	\begin{align*}
		\dc(p) \geq \sqrt{\brank(p_{x_0}^{(2k)})}. 
		%\left(k \leq D,\  x_0 \in {\rm Zeros}(p) \right). 
	\end{align*}
\end{prop}
We omit the case $x_0 \in {\rm Zeros}(p)$ for the simplicity of the proof.
We use the following lemma which is a variation of Lemma~\ref{lem_key1}.
\begin{lem}\label{lem_key_appendix}
	Let $p \in K[\xx]$ with $\dc (p)=n$. Then, for all $x_0 \in K^D \setminus {\rm Zeros}(p)$, there exist a linear polynomial matrix $A \in (K[\xx])^{n \times n}$ and $\alpha \in \RR$
	such that $p_{x_0}(x) = \alpha \det(A(x) + I)$.
	\begin{proof}
		From the definition of the determinantal complexity, there exists an affine polynomial matrix 
		$Q\in (K[\xx])^{n \times n}$ such that $p = {\det} ( Q)$. 
		In particular, given any $x_0 \in K^D \setminus {\rm Zeros}(p)$ it follows that $p_{x_0}(x) = {\det} (Q(x + x_0))$.
		Define $\alpha := \det (Q(x_0)) = p_{x_0}(0) \neq 0$.
		Since $Q$ is an affine polynomial matrix, we can represent $Q(x+x_0) = L(x) + Q(x_0)$ where $L\in (K[\xx])^{n \times n}$ is a linear polynomial matrix.
		We have 
		\[
		{\det}(Q(x +x_0))= \alpha \det(Q(x_0)^{-1}) {\det}(L(x) + Q(x_0)) = \alpha {\det} (Q(x_0)^{-1}L(x) +  I) .
		\]
		Since $Q(x_0)^{-1}L(x)$ is also a linear polynomial matrix, 
		we define $A(x) := Q(x_0)^{-1}L(x)$, and then the statement follows.
	\end{proof}
\end{lem}
We formally define an ABP discussed in Section~\ref{sec_intro}.
\begin{df}[Nisan~\cite{Nisan91lowerbounds}, see also~\cite{Shpilka2009ACa}]
	An (homogeneous) arithmetic branching program (ABP) over $K[x]$ is a layered graph with $n+1$ layers
	as follows.
	The layers are labeled by $0,1,\dots,n$. The edges of the graph go from layer $i$ to layer $i+1$.
	Every edge $e$ is labeled by  a (homogeneous) linear polynomial $\ell_e \in K[x]$.
	Layer $0$ has only one vertex called the source, and layer $n$ has only one vertex called the sink.
	For every directed path from the source to the sink $\gamma = (e_1,e_2,\dots,e_n)$,
	define the polynomial $f_{\gamma}$ associated to $\gamma$ as $f_{\gamma} = \ell_{e_1}\ell_{e_2}\cdots \ell_{e_n}$. The polynomial computed by ABP is $\sum_{\gamma} f_{\gamma}$.
	
	For an ABP $\cA$,
	we define the width $w_k(\cA)$ of layer $k$ as the number of vertices in the layer $k$.
	Given a homogeneous polynomial $f$ with degree at least $k$, we denote by $w_k(f)$ the minimum $w_k(\cA)$ over ABPs $\cA$ which compute $f$.
\end{df}
The following is an easy observation.
\begin{fct}\label{fct_ABP_brank}
	%(i) If homogeneous polynomials $f,g,h$ of degree at least $k$ satisfy $f=g+h$,
	%then $w_k(f) \leq w_k(g) + w_k(h)$.
	%(ii) 
	For $f \in K[x]^{(2k)}$, it holds that $\brank(f) \leq w_k(f)$.
	\begin{proof}
		%(i) is immediate by definition.
		%For (ii), 
		Suppose that an ABP $\cA$ computes $f$.
		Let $V$ be the set of vertices in the layer $k$ of $\cA$.
		For $v \in V$,
		let $\mathcal{R}_{v}$ and $\mathcal{R}_{v}'$ be the sets of path from the source to $v$ and 
		from $v$ to the sink, respectively.
		Then it holds that 
		\begin{align*}
			f = \sum_{v \in V} \left(\sum_{\gamma \in \mathcal{R}_{v}} f_{\gamma} \right) \left( \sum_{\gamma' \in \mathcal{R}_{v}} f_{\gamma'}  \right).
		\end{align*}
		Therefore we have $\brank(f) \leq |V| = w_k(f)$.
	\end{proof}
\end{fct}
The next statement is a well-known result.
\begin{thm}[Mahajan and Vinay~\cite{Mahajan1997Dca}]\label{thm_characteristic_poly}
	Let $A \in (K[\xx])^{n \times n}$ be a linear polynomial matrix, and $r \in [1,n-2]$.
	Then there exists an ABP $\cA$ over $K[x]$ 
	such that $\cA$ computes the coefficient of $\lambda^{n-r}$ in $\det(A(x) + \lambda I)$
	and satisfying $w_k (\cA) \leq n^2$ for all $k \in [1,r-1]$.
\end{thm}
Then Proposition~\ref{prop_sqrt} is proved as follows.
\begin{proof}[Proof of Proposition~\ref{prop_sqrt}]
	By Lemma~\ref{lem_key_appendix},
	there exists a linear polynomial matrix $A \in (K[\xx])^{n \times n}$ and $\alpha \in \RR$
	such that $p_{x_0}(x) = \alpha \det(A(x) + I)$.
	Then $p_{x_0}^{(2k)}$ is equal to the coefficient of $\lambda^{n-2k}$ in $ \alpha \det(A(x) + \lambda I)$.
	Then by Theorem~\ref{thm_characteristic_poly},
	there exists an ABP $\cA$ with $w_k(\cA) \leq n^2$ which computes $p_{x_0}(x)^{(2k)}$.
	%Let $V$ be the set of vertices in the layer $k$ of $\cA$.
	%For $v \in V$,
	%let $\mathcal{R}_{v}$ and $\mathcal{R}_{v}'$ be the sets of path from the source to $v$ and 
	%from $v$ to the sink, respectively.
	%Then it holds that 
	%\begin{align*}
	%	p_{x_0}^{(2k)} = \sum_{v \in V} \left(\sum_{\gamma \in \mathcal{R}_{v}} f_{\gamma} \right) \left( \sum_{\gamma' \in \mathcal{R}_{v}} f_{\gamma'}  \right).
	%\end{align*}
	By Fact~\ref{fct_ABP_brank}, we have $\brank(p_{x_0}^{(2k)}) \leq w_k(p_{x_0}^{(2k)}) \leq w_k(\cA) \leq n^2 = \dc(p) ^2$.
\end{proof}

%% file: brank_bound.tex
\section{Lower bounds of $\dc(\perm_d)$ by b-rank: case $k=1$}\label{sec_relation_b-rank_detcomp}
Considering the case $k=1$ in Theorem~\ref{thm_main_dc},
we obtain lower bounds of $\dc(\perm_d)$ by the bi-polynomial rank.
We define $\Sigma_d \in {\rm Zeros}(\perm_{d})$ as follows:
\begin{align*}
\Sigma_d := \left( 
\begin{array}{cccc}
1 &\cdots &\cdots &1 \\
\vdots &\ddots &\ddots &\vdots \\
\vdots &\ddots &1 &1\\
1 &\cdots &1 &1-d 
\end{array}
\right).
\end{align*}
This is the same matrix appearing in the proof of 
Theorem~\ref{thm_mignon} in~\cite{Mignon2004Aqb}.

\subsection{Mignon-Ressayre bound from b-rank}\label{sec_demonstration}
This section is devoted to the demonstration of the bi-polynomial rank
as an extention of Hessian rank.
%In this section, we demonstrate how the bi-polynomial rank inherits the information of differentials, 
We give an alternative proof of the result of Mignon and Ressayre (Theorem~\ref{thm_mignon}).
By using the bi-polynomial rank, Theorem~\ref{thm_mignon} immediately follows from Theorem~\ref{thm_main_dc} as follows.
\begin{proof}[An alternative proof of Theorem~\ref{thm_mignon}]
	For any $p \in K[x]$ and $x_0 \in {\rm Zeros}(p)$, we have
	%In the case $k=1$, one has 
	\[
	p_{x_0}^{(2)}(x) = \frac{1}{2} 
	\sum_{1 \leq i,j \leq  D}  x_i x_j
	\left.
	\left( \frac{\partial^2}{\partial x_i \partial x_j} p  \right)
	\right|_{x=x_0} .
	\]
	We define the Hessian $H_{p,x_0}=(h_{i,j})$ of $p$ at $x_0$ by $h_{i,j} := \left.
	\left( \frac{\partial^2}{\partial x_i \partial x_j} p  \right)
	\right|_{x=x_0}$.
	By definition, $\brank(p_{x_0}^{(2)})$ is equal to the minimum number $n$ of bilinear forms $(\sum_{l=1}^{D} b^m_{l}x_l )(\sum_{l=1}^{D} c^m_{l}x_l )$
	$(m=1,2,\dots,n)$ whose sum is equal to $p_{x_0}^{(2)}$.
	We define the rank one matrices $A_m = (a^m_{i,j})$ for $m=1,2,\dots,n$, by $a^m_{i,j} :=b^m_{i} c^m_{j} $.
	Let $A := \sum_{m=1}^{n} A_m$, and then it holds that $A + A^{\top} = H_{p,x_0}$. 
	Therefore we have $\brank(p_{x_0}^{(2)}) \geq \rank(A) \geq \frac{1}{2} \rank(H_{p,x_0})$.
	By Theorem~\ref{thm_mignon}, by putting $k=1$ it holds that
	\[
	\dc(p) \geq \brank(p_{x_0}^{(2)}) \geq \frac{1}{2} \rank(H_{p,x_0}).
	\]
	In the case of $p = \perm_d$, Mignon and Ressayre 
	proved $\rank(H_{\perm_d,\Sigma_d}) = d^2$.
	Thus Theorem~\ref{thm_mignon} follows from our Theorem~\ref{thm_main_dc}.
\end{proof}

\subsection{Lower bound of $\dc(\perm_d)$ over the field $\RR$}\label{subsec_newbound}
Theorem~\ref{thm_d2bound} improves the current best lower bound given by Mignon and Ressayre.
We present the proof in this section.
%giving a new bound of $\dc(\perm_d)$ on real field.
%From Theorem \ref{thm_main_dc} and Corollary \ref{cor_to_opt}, a slightly weaker version of Theorem \ref{thm_mignon} can be proved as follows.
Given a symmetric matrix, we denote by a tuple $(n_+, n_-, n_0)$ the {\em signature} of the matrix, that is, the number of positive, negative, zero eigenvalues, respectively.
If symmetric matrices $S$ and $S'$ have the same signature, we denote $S \sim S'$.
By Sylvester's law of inertia, $S \sim S'$ if and only if
$S \sim T S T^{\top}$ 
for a nonsingular matrix $T$.
We use the next lemma.
\begin{lem}\label{lem_signature}
	Let $Q \in \Mat_n(\RR)$, $Q_{{\rm sym}} := Q + Q^{\top}$,
	and $(n_+, n_-, n_0)$ be the signature of $Q_{{\rm sym}}$.
	Then it holds that $\rank(Q) \geq \max\{n_+, n_-\}$.
	\begin{proof}
		Let $\lambda_1, \lambda_2,\dots,\lambda_{n_+}$ be the positive eigenvalues of $Q_{{\rm sym}}$, 
		and $\vv_1, \vv_2, \dots,\vv_{n_+}$ be the corresponding eigenvectors which are orthogonal to each other.
		Let $V_+ \subseteq \RR^n$ be the $n_+$-dimensional subspace spanned by $\vv_1, \vv_2, \dots,\vv_{n_+}$.
		For any nonzero vector $\uu = \sum_{i=1}^{n_+} a_i \vv_i \in V_+$ where $a_1,a_2,\dots,a_{n_+}$ $ \in \RR$,
		it holds that
		\[
		2\uu^{\top} Q \uu = \uu^{\top} Q_{{\rm sym}} \uu = \sum_{i=1}^{n_+} a_i^2 \vv_i^{\top} Q \vv_i =  \sum_{i=1}^{n_+} a_i^2 \lambda_i > 0.
		\]
		The above inequality shows that $Q \uu \neq 0$ for all nonzero vectors $\uu$ in $n_+$-dimensional space $ V_+$, and therefore $\rank (Q) \geq n_+$.
		The same argument is also true for eigenvectors with negative eigenvalues, 
		and it holds that $\rank (Q) \geq n_-$.
	\end{proof} 
\end{lem}
The proof of Theorem~\ref{thm_d2bound} is given as follows.
\begin{proof}[Proof of Theorem~\ref{thm_d2bound}]
	%We consider $\perm_d \in \RR[x_{1,1},x_{1,2}, \dots, x_{n,n-1}, x_{n,n}]$.
	%Since $\perm_d$ is a polynomial on $d^2$ variables, we set $D := d^2$.
	%The number of pairs $(i,j)$ of integers $i,j \in d$ is $D$, and we consider that an element of $I=(I_{i,j}) \in \mathcal{I}_{D}^{(k)}$ is indexed by a pair of integers $i,j \in [d]$.
	%Then we also denote $x^I = \prod_{i,j \in [d]} x_{i,j}^{I_{i,j}}$.
	%We regard a pair of indices $(i,j)$ as an element $I$ of $\mathcal{I}_{D}^{(1)}$,
	%by defining $I_{di + j} = 1$.
	%
	%Let $K$ be $\RR$ or $\CC$
	From Theorem \ref{thm_main_dc}, for $k=1$ we obtain that $\dc(\perm_d) \geq \brank(\perm_{d,\Sigma_d}^{(2)})$ over $\RR$.
	We consider that a matrix $A = (a_{(i,j),(i',j')})$ of size $d^2$ is indexed by pairs of integers $(i,j), (i',j')$ where $i,j,i',j' \in [1,d]$.
	Since 
	\[
	{\perm}_{d,\Sigma_d}^{(2)}(x) = \frac{1}{2}\sum_{i,j,i',j' \in [1,d]} x_{i,j}x_{i',j'} \left.
	\left( \frac{\partial^2 {\perm}_d}{\partial x_{i,j} \partial x_{i',j'}} \right) 
	\right|_{x=\Sigma_d} ,
	\] 
	we define the Hessian matrix $H =(h_{(i,j),(i',j')})$ of $\perm_d$ at $\Sigma_d$ by the following equation.
	\[
	h_{(i,j),(i',j')} = h_{(i',j'),(i,j)} :=  \left.
	\left( \frac{\partial^2 {\perm}_d}{\partial x_{i,j} \partial x_{i',j'}} \right) 
	\right|_{x=\Sigma_d}.
	\]
	The corresponding optimization problem in Theorem~\ref{thm_to_opt} is equal to the following.
	\begin{align*}
	\begin{array}{ll}
	{\rm Minimize} \quad &\rank(Q) \\
	{\rm subject \ to} \quad & 
	Q + Q^{\top} = H,\\
	& Q \in \Mat_{d^2}(\RR).
	\end{array}
	\end{align*}
	%where nonzero elements of $B_{i,j} \in M_{d^2}(\CC)$ are $(B_{i,j})_{i,j} = (B_{i,j})_{j,i} = 1$.
	Let $Q_{{\rm opt}}$ be an optimum solution of the above problem. 
	By Theorem~\ref{thm_to_opt} we have $ \brank({\perm}_{d,\Sigma_d}^{(2)}) = \rank(Q_{{\rm opt}})$.
	%Since every $B_{H}$ is symmetric and has at most two nonzero elements in symmetric position,
	%for every feasible solution $Q$, the symmetrized matrix $(Q + Q^{\top})$ is constant. 
	Let $(n_+, n_-, n_0)$ be the signature of $H \in \Sym_{d^2}$.
	Since $Q_{{\rm opt}} + Q_{{\rm opt}}^{\top} =H$ by Lemma~\ref{lem_signature} it follows that $\rank(Q_{{\rm opt}}) \geq n_-$.
	Therefore we obtain
	\[
	\dc({\perm}_d) \geq \brank({\perm}_{d,\Sigma_d}^{(2)})  = \rank(Q_{{\rm opt}}) 
	\geq n_-. 
	\]
	
	We are going to prove $n_- = (d-1)^2 + 1$.
	As in~\cite{Mignon2004Aqb}, $H$ can be calculated as
	\[
	H = (d-3)!
	\left( 
	\begin{array}{cccccc}
	O& B& \cdots& B& C\\
	B& \ddots& \ddots& \vdots& \vdots \\
	\vdots& \ddots& O& B& C\\
	B& \cdots& B& O& C\\ 
	C& \cdots& C& C& O\\ 
	\end{array}
	\right),
	\]
	where $B$ and $C$ are the following matrices of size $d$:
	\[
	B=
	\left(
	\begin{array}{ccccc}
	0& -2& \cdots& -2& d-2\\
	-2& \ddots&  \ddots& \vdots& \vdots\\
	\vdots& \ddots&  0& -2& d-2\\
	-2& \cdots& -2& 0& d-2\\
	d-2& \cdots& d-2& d-2& 0
	\end{array} 
	\right) ,
	C=(d-2)
	\left(
	\begin{array}{cccc}
	0& 1& \cdots& 1\\
	1& 0& \ddots& \vdots\\
	\vdots&  \ddots&  \ddots& 1\\
	1& \cdots& 1& 0
	\end{array}
	\right)
	. \ 
	\]
	Let $S_d$ be the symmetric matrix of size $d$ defined as follows.
	\[
	S_d :=
	\left(
	\begin{array}{cccc}
	0& 1& \cdots& 1\\
	1& 0& \ddots& \vdots\\
	\vdots&  \ddots&  \ddots& 1\\
	1& \cdots& 1& 0
	\end{array}
	\right)
	.
	\]
	Then it holds that $B \sim -S_d$ and $C \sim S_d$.
	Let $I_d$ be the identity matrix of size $d$.
	The signature of $S_d$ is $(1,d-1,0)$,
	since the rank of the matrix $(S_d + I_d)$ is one with the nonzero eigenvalue $d$.
	Define a nonsingular matrix $T$ of size $d^2$ as follows.
	\[
	T:=
	\left(
	\begin{array}{ccccc}
	
	I_d& O& \cdots& \cdots& O\\
	O& I_d&  \ddots& \ddots& \vdots\\
	\vdots& \ddots&  \ddots& \ddots& \vdots\\
	O& O& \cdots& I_d& O\\
	-\frac{1}{d-2}CB^{-1}& -\frac{1}{d-2}CB^{-1}& \cdots& -\frac{1}{d-2}CB^{-1}& I_d
	\end{array} 
	\right) .
	\]
	Then $H \sim T H T^{\top}$, where
	\[
	TQ_{{\rm sym}} T^{\top}=
	\left(
	\begin{array}{ccccc}
	O& B& \cdots& O& O\\
	B& O& \ddots& \vdots& \vdots \\
	\vdots& \ddots& \ddots& B& O\\
	B& \cdots& B& O& O \\
	O& \cdots& \cdots& O& -\frac{d-1}{d-2}CB^{-1}C\\
	\end{array} 
	\right) .
	\]
	Let $H'$ be the upper-left principal submatrix of $T H T^{\top}$ of size $d(d-1)$,
	which is represented as follows. 
	\[
	Q' =
	\left(
	\begin{array}{cccc}
	O& B& \cdots& B\\
	B& O& \ddots& \vdots\\
	\vdots& \ddots& \ddots& B\\
	B& \cdots& C& O 
	\end{array} 
	\right) .
	\]
	Denote by $(l_+, l_-, l_0)$ and $(m_+, m_-, m_0)$ the signatures of $(-CB^{-1}C)$ and $H'$, respectively. Then it holds that $n_- = l_- + m_-$.
	Since $-CB^{-1}C \sim - (BC^{-1}) (CBC)(BC^{-1})^{\top} = -B \sim S_d$, 
	$l_-$ is equal to $d-1$.
	On the other hand, it holds that $H' = S_{d-1} \otimes C$, where $\otimes$ denote the Kronecker product.
	Since the set of eigenvalues of $S_{d-1} \otimes B$ consists of the products of all pair of eigenvalues of $S_{d-1}$ and $B$, the signature of $S_{d-1} \otimes B$ is $(2d - 3, (d-2)(d-1)+1,0)$. Therefore we have $n_- = (d-1) + ((d-2)(d-1) + 1) = (d-1)^{2} + 1$. 
\end{proof}

%% file: brank_conjecture.tex
\section{Toward strong lower bounds via concave minimization}\label{sec_cnj}
We formulate the bi-polynomial rank as 
the minimum matrix-rank over an affine subspace of matrices in Section~\ref{section_CB}.
Unfortunately, few results are known 
for giving theoretical lower bounds for the rank minimization problem.
In our case, the calculation of such a minimum rank is
still difficult for $k \geq 2$.
We propose an approach to bound the minimum rank below
by using the framework of the concave minimization.
In this section, we fix the field $K = \RR$.
\subsection{Concave minimization for bounding minimum rank below}\label{subsec_Cooncavemin}
The object of the concave minimization is
to minimize a concave function over a convex set.
This setting is studied in the area of global optimization~\cite{Pardalos1986Mfg}.
We use this framework to obtain lower bounds of the minimum rank 
over a subset of positive semidefinite matrices.

%Recall the definition of $\lambda_{[l]}$ in Section~\ref{sec_intro}.
As in Section~\ref{sec_intro}, for $Y \in \Sym_n$ and $l \in [1, n]$, we denote by $\mu_{l}(Y)$ the sum of the smallest $l$ eigenvalues of $Y$.
For $\cX \subseteq \Sym_n$,
we define $\minrank(\cX) := \min_{X \in \cX} \rank(X)$.
Then the next statement is immediate from the definition of positive semidefinite matrices.
\begin{lem}\label{lem_rank_eigen}
	Let $\cX \subseteq \Psd_n$, and $r \in \NN$.
	Then $\minrank(\cX) > r$ if and only if 
	$\mu_{n-r}(X) > 0$ for all $X \in \cX$.
\end{lem}

This statement suggests the way to solve the rank minimization problem
over positive semidefinite matrices
by minimizing $\mu_{l}$ over the feasible region.
%For 
%$X = (x_{i,j}),Y = (y_{i,j}) \in \Sym_n$, define $X \cdot Y := \sum_{i=1}^{n} \sum_{j=1}^{n} x_{i,j} y_{i,j}$.
The computation of $\mu_l$ is formulated as
the optimum solution of a semidefinite programming.
\begin{prop}[{See~\cite[Section 4.1]{Alizadeh1995Ipm}}]\label{prop_eigens_sdp}
	Let $A \in \Sym_n$ and $l \in [1,n]$.
	Then $\mu_{l}(A)$ is equal to the optimum value of the following 
	problem:
	\begin{align*}
	\begin{array}{ll}
	{\rm Minimize} \quad & \tr(AX) \\
	{\rm subject \ to} \quad &  \tr(X) = l, \\
	& X, I-X \in \Psd_{n}.
	\end{array}
	\end{align*}	
	\begin{proof}	
		Since $A$ is a symmetric matrix, $A$ is diagonalizable by some orthogonal matrix $V$,
		as $V A V^{\top} =: \tilde{A}$.
		Let $\tilde{X} := V X V^{\top}$,
		and then the replacement of $A,X$ by $\tilde{A}, \tilde{X}$ does not change the optimum value. Therefore, without loss of generality, we can assume that $A$ is a diagonal matrix with diagonal entries $\lambda_1 \leq  \lambda_2 \leq \dots \leq \lambda_n$.
		Then the optimum value is attained by such $X$ that the first $l$ diagonal entries are $1$, and the other entries are $0$.
		It holds that the optimum value is $\sum_{i=1}^l \lambda_i =  \mu_{l}(A)$.
	\end{proof} 
\end{prop}
For latter use, we prepare the next statement.
\begin{cor}\label{cor_dual}
	Let $Y \in \Sym_n$. Then $\mu_l (Y)$ is at least $l z - \tr(Z)$
	for $Z \in \Psd_n$ and $z \in \RR$ satisfying $(Y + Z - zI) \in \Psd_n$.
	\begin{proof}
		Observe that the following optimization problem is the dual of the problem in Proposition~\ref{prop_eigens_sdp}.
		\begin{align*}
		\begin{array}{ll}
		{\rm Maximize} \quad & lz - \tr(Z) \\
		{\rm subject \ to} \quad &  z \in \RR, \\
		& Z, Y + Z - zI \in \Psd_{n}.
		\end{array}
		\end{align*}
		By the weak duality of semidefinite programming, for any feasible solution $(Z,z) \in \Psd_n \times \RR$,
		the objective value $l z - \tr(Z)$ is at most $\mu_l (Y)$.
	\end{proof}
\end{cor}
The concavity of $\mu_{l}$ is proved as follows.
\begin{cor}[{See~\cite[Section 4.1]{Alizadeh1995Ipm}}]\label{cor_eigens_concave}
	For $l \in [1,n]$, $\mu_{l}: \Sym_n \to \RR$ is a concave function.
	\begin{proof}
		Let $X, Y \in \Sym_n$. Denote by $\prob(Y)$ the optimization problem in Proposition~\ref{prop_eigens_sdp} for $Y$.
		Then an optimum solution of $\prob(\frac{X+Y}{2})$ is a feasible solution of both $\prob(X)$ and $\prob(Y)$,
		and therefore the optimum value of $\prob(\frac{X+Y}{2})$ is at least
		the average of optimum values of $\prob(X)$ and $\prob(Y)$.
		By Proposition~\ref{prop_eigens_sdp}, this indicates that
		$\frac{1}{2}(\mu_l(X) + \mu_l(Y)) \leq  \mu_l(\frac{X+Y}{2})$.
	\end{proof}
\end{cor}
In the theory of the concave minimization,
the outer approximation approach~\cite{Kelley1960tcp,Tuy1983Ooa} (see also~\cite{Pardalos1986Mfg})
obtain a lower bound of the minimum of a given concave function
by approximating the feasible region from outside.
Given a set $\cY \subseteq \Sym_n$,
we denote by $\conv(\cY)$ the convex hull of $\cY$.
In general, given a concave function $f$ over $\cY$,
it holds that $\min_{Y \in \cY} f(Y) = \min_{Y \in \conv(\cY)} f(Y)$.
Therefore the next statement holds.
\begin{thm}\label{thm_concave_bound}
	Let $\cX \subseteq \Psd_n$ and $r \in \NN$. 
	If there exists $\cY \subseteq \Sym_n$ satisfying
	the following property, then $\minrank(\cX) > r$.
	
	(i) $\cX \subseteq \conv(\cY)$.
	 
	(ii) $\mu_{n-r}(Y) > 0$ for all $Y \in \cY$.
	\begin{proof}
		Since $\mu_{n-r}$ is a concave function,
		$0 < \min_{Y \in \cY} \mu_{n-r}(Y) = \min_{Y \in \conv(\cY)} \mu_{n-r}(Y)
		 \leq  \min_{X \in \cX} \mu_{n-r}(X)$.
		By Lemma~\ref{lem_rank_eigen}, 
		$\mu_{n-r}(X) >0$ for all $X \in \cX$ implies $\minrank(\cX) > r$.
	\end{proof}
\end{thm}
%We can choose any set $\cY$ satisfying $\cX \subseteq \conv(\cY)$,
%and therefore the calculation of $\mineigen_{n-r}(\cY)$
%may be tractable compared to the calculation of $\minrank(\cX)$.
%

%Corollary~\ref{cor_dual} may help
%to verify $\mu_{n-r}(Y) >0$ theoretically,
%as follows.
%\begin{cor}\label{cor_concave_bound}
%	Let $\cX \subseteq \Psd_n$ and $r \in \NN$. 
%	If there exists $\cY \subseteq \Sym_n$ satisfying
%	the following property, then $\minrank(\cX) > r$.
%	
%	(i) $\cX \subseteq \conv(\cY)$.
%	
%	(ii) For all $Y \in \cY$, there exists $(Z,z) \in \Psd_n \times \RR$ such that 
%	$Y + Z - zI \in \Psd_n$ and $(n-r)z - \tr(Z) > 0$.
%	\begin{proof}
%		The statement directly follows from Theorem~\ref{thm_concave_bound} and 
%		Corollary~\ref{cor_dual}.
%	\end{proof}
%\end{cor}
%
%
To utilize Theorem~\ref{thm_concave_bound} for lower bounds of the bi-polynomial rank,
we prove Proposition~\ref{prop_concave_bound}.
\begin{proof}[Proof of Proposition~\ref{prop_concave_bound}]
	By Corollary~\ref{cor_brank_psd_rankmin}, $\brank(p) \geq \frac{1}{2} \minrank(\cX_p \cap \Psd_{n})$.
	Then by Theorem~\ref{thm_concave_bound},
	$\minrank(\cX_p \cap \Psd_{n}) > r$, and therefore the statement holds.
\end{proof}
Corollary~\ref{cor_dual} may help
the verification $\mu_{n-r}(Y) >0$
as follows.
\begin{cor}\label{cor_conjecture}
	Let $p \in K[x]^{(2k)}$ and $r \in \NN$.
	If there exists $\cY \subseteq \cX_p$ satisfying the following property, then $\brank(p) > r /2$.
	
	(i) $\cX_p \cap \Psd_{2s_k} \subseteq \conv(\cY)$.
	
	(ii) For all $(Y_1,Y_2) \in \cY$, 
	there exists $(Z_1,Z_2,z) \in \Psd_{s_k} \times \Psd_{s_k} \times \RR$ 
	such that $(Y_1 + Z_1 - zI), (Y_2 + Z_2 - zI) \in \Psd_{s_k}$
	and $ (2 s_k - r) z - \tr(Z_1 + Z_2) > 0$.
	\begin{proof}
%		By Corollary~\ref{cor_brank_psd_rankmin}, $\brank(p) \geq \frac{1}{2} \minrank(\cX_p \cap \Psd_{2s_k})$.
%		Then by Corollary~\ref{cor_concave_bound},
%		$\minrank(\cX_p \cap \Psd_{2s_k}) > r$, and therefore the statement holds.
		The statement directly follows from Proposition~\ref{prop_concave_bound} and 
		Corollary~\ref{cor_dual}.
	\end{proof}
\end{cor}

\subsection{An explicit representation of $\cX_{\perm_{d,\Sigma_d}^{(2k)}}$}\label{subsec_explicit}
The previous section discusses a general framework 
for giving lower bounds of the bi-polynomial rank.
%In the framework, to calculate $\brank(p)$ of $p \in K[x]^{(2k)}$,
%we first find $\cY \subseteq \cX_p$ satisfying $\cX_p \cap \Psd_{2s_k} \subseteq \conv(\cY)$,
%and then verify $\mu_{n-r}(Y)>0$ for all $Y \in \cY$.
In the framework, to calculate $\brank(p)$ of $p \in K[x]^{(2k)}$, 
we consider the minimum of the concave function $\mu_{2s_k - r}$ over $\cX_p \cap \Psd_{2s_k}$.
Our final target is to obtain lower bounds of $\dc(\perm_d)$.
In this section,
we fix $p = \perm_{d,\Sigma_d}^{(2k)}$,
and give an explicit representation of a projection $\cZ_{2k}$ of $\cX_{p}$.
$p$ is a multilinear polynomial,
and to extract this feature,
we define the subset ${\cJ_{k,d^2}} := \cI_{k,d^2} \cap \{0,1 \}^{d^2}$ 
of $d^2$-tuples.
Define $d' := d-1$.
Observe that 
$\Sigma_d$ has good symmetry except for $d$th row/column.
For $I = (i_{1,1},i_{1,2},\dots,i_{d',d'})  \in \cJ_{k,{d'}^2}$,
we define $\iota(I) \in \cJ_{k,d^2}$ by the insertion of zeros
into the entries not in $\cJ_{k,{d'}^2}$.
More concretely,
\[
\iota(i_{1,1},i_{1,2},\dots,i_{d',d'}) := (i_{1,1},\dots,i_{1,d-1},0,i_{2,1},\dots,i_{d',d'},0,\dots,0).
\]
Let $t_k$ be the cardinality of $\cJ_{k,{d'}^2}$.
Then we define the projection $\pi: \Sym_{s_k} \to \Sym_{t_k}$ by 
$\pi(Y)_{I,J} := \alpha Y_{\iota(I), \iota(J)}$ for $I,J \in \cJ_{k,{d'}^2}$,
where $\alpha := - \frac{1}{2k(d-2k-1)!}$ is a constant.
We define $\cZ_{2k} \subseteq \Sym_{t_k} \times \Sym_{t_k}$ as a projection of $\cX_{p}$ as follows:
\begin{align*}
	\cZ_{2k} := \{ (\pi(X_+), \pi(X_-))  \mid (X_+,X_-) \in \cX_{p}  \}.
\end{align*}
Then it holds that 
\begin{align*}
\minrank(\cZ_{2k} \cap \Psd_{2t_k}) 
&\leq \minrank(\cX_{p} \cap \Psd_{2s_k}) \\
&\leq \minrank(\cZ_{2k} \cap \Psd_{2t_k}) + (s_k - t_k),
\end{align*}
where $s_k - t_k = O(d^{2k-2})$.
Hence $\minrank(\cX_{p} \cap \Psd_{2s_k}) = \Omega(d^{2k})$
if and only if $\minrank(\cZ_{2k} \cap \Psd_{2t_k}) = \Omega(d^{2k})$.
We are going to give an explicit representation of $\cZ_{2k}$.
In $p$,
no monomial with repetition of a row/column index appear.
To express this,
we classify $\cJ_{2k,{d'}^2}$ into $\cH_1$ and $\cH_0$ as follows.
We define $\cH_1$ as the set of tuples 
$(H_{1,1},H_{1,2},\dots,H_{d',d'}) \in \cJ_{2k,{d'}^2}$
satisfying $\sum_{j' =1}^{d'} H_{i,j'} \leq 1$ and $\sum_{i' =1}^{d'} H_{i',j} \leq 1$
for $i,j \in [1,d']$,
and $\cH_0 := \cJ_{2k,{d'}^2} \setminus \cH_1$.
%\begin{align}
%\cH_1 &:= \left\{ (H_{1,1},H_{1,2},\dots,H_{d,d-1},H_{d,d}) \in {\cJ_k}^{(d-1)^2} \left| 
%\substack{%
%	\sum_{j' =1}^{d} H_{i,j'} \leq 1, \ i \in [1,d], \\
%	\sum_{i' =1}^{d} H_{i',j} \leq 1, \ j \in [1,d]
%}
%\right\} \right. ,\label{eq_double1} \\
%\cH}_0 &:= \binom{d^2}{2k} \setminus {\cJ_k}^{(d-1)^2}. \label{eq_double2}
%\end{align}
%
Then $\cZ_{2k}$ is given as the set of pairs $(U,V)$ of matrices in $\Sym_{t_k}$
satisfying linear equations for all $H \in \cJ_{2k,{d'}^2}$:
\begin{align*}
\sum_{\substack{I,J \in \cJ_{k,{d'}^2} \\ I + J = H }} (u_{I,J} - v_{I,J}) = 
\begin{cases}
1 , \ H \in \cH_1 \\ 
0 , \ H \in \cH_0
\end{cases}
\end{align*}
Observe that this affine space $\cZ_{2k}$
is represented by simple linear equations
with coefficients in $\{0, \pm 1 \}$.

%% file: brank_proof1.tex
\section{Proof of Theorem~\ref{thm_main_dc}}\label{sec_proof}
A {\em linear polynomial matrix} over $x_1,x_2,\dots,x_D$ of size $n$ is an $n \times n$ matrix $A(x) = (a_{i,j}(x)) \in (K[\xx])^{n \times n}$, where each element $a_{i,j}(x)$ is a (homogeneous) linear polynomial for $1 \leq i,j \leq n$.
Denote by $\Lambda_n^r$ the diagonal matrix of size $n$ with diagonal entries $(0,\dots,0,\underbrace{1,\dots,1}_{r})$. 
For $\alpha,\beta \in \ZZ$ with $\alpha \leq \beta$,
we denote $\{ \alpha,\alpha+1,\dots,\beta \}$ by $[\alpha,\beta]$.
\begin{lem}\label{lem_key1}
Let $p \in K[\xx]$ with $\dc (p)=n$. Then, for all $x_0 \in {\rm Zeros}(p)$, there exist a linear polynomial matrix $A \in (K[\xx])^{n \times n}$ and $r \in [0,n-1]$
such that $p_{x_0}(x) = \det(A(x) + \Lambda^r_n)$.
\begin{proof}
From the definition of the determinantal complexity, there exists an affine polynomial matrix 
$Q\in (K[\xx])^{n \times n}$ such that $p = {\det} ( Q)$. 
In particular, given any $x_0 \in {\rm Zeros}(p)$ it follows that $p_{x_0}(x) = {\det} (Q(x + x_0))$.
Observe that $\det (Q(x_0)) = p_{x_0}(0) = 0$ and $\rank(Q(x_0)) \leq n-1$. 
Let $r\in [0,n-1]$ be the rank of $Q(x_0)$. 
Then there exist nonsingular matrices $S,T$ of size $n$ such that $S Q(x_0) T = \Lambda_n^r$ and $\det(ST)=1$.
Since $Q$ is an affine polynomial matrix, we can represent $Q(x+x_0) = L(x) + Q(x_0)$ where $L\in (K[\xx])^{n \times n}$ is a linear polynomial matrix.
We have 
\[
{\det}(Q(x +x_0))= {\det}(L(x) + Q(x_0)) = {\det} (S(L(x) + Q(x_0))T) = {\det}(S L(x)T + \Lambda_n^r).
\]
Since $S L(x)T$ is also a linear polynomial matrix, we define $A(x) := S L(x)T$, and then the statement follows.
\end{proof}
\end{lem}
Given a linear polynomial matrix $A \in (K[\xx])^{n \times n}$, $k \in \NN$ and $r \in [0,n-1]$,
we define 
\[
p_{A,k,r}(x) := ({\det}( A(x) + \Lambda^r_n))^{(k)}.
\]
Then it holds that 
${\det}( A(x) + \Lambda^r_n) = \sum_{k=n-r}^n p_{A,k,r}(x)$.
The next statement is the essence of our result.
\begin{prop}\label{prop_ABP_p_n-1}
	For $A \in (K[\xx])^{n \times n}$,
	it holds that $\brank(p_{A,2k,n-1}) \leq n + 2(k-1) D^{k-1}$.
\end{prop}
The proof of Proposition~\ref{prop_ABP_p_n-1} is given in the next section.
By this proposition, the following statement holds.
\begin{lem}\label{lem_ABP_p}
	(i) For $r \in [n-2k, n-1]$, it holds that $\brank(p_{A,2k,r}) \leq 2^{n - r -1}(n + 2(k-1)D^{k-1})$.
	
	(ii) For $r = n-2k$, it holds that $\brank(p_{A,2k,r}) \leq \binom{2k}{k}$.
\begin{proof}
	(i), let $i := n - r$.
	We prove the statement by the induction on $i$.
	The case $i=1$ directly follows from Proposition~\ref{prop_ABP_p_n-1}.
	Suppose that for $i \leq 2k-1$, the statement holds.
	Denote by $A' \in K[x]^{n-1 \times n-1}$ the linear polynomial matrix obtained by deleting the $(i+1)$th row and column of $A$.
	Since determinant is a bi-linear form, we have
	\begin{align*}
		\det(A(x) + \Lambda_n^{n-i}) = \det(A(x)+\Lambda_n^{n-(i+1)}) + \det(A'(x) + \Lambda_{n-1}^{(n-1)-i}),
	\end{align*}
	and therefore $p_{A,2k,n-(i+1)} = p_{A,2k,n-i} - p_{A', 2k, (n-1) - i}$.
	By inductive hypothesis, we have
	\begin{align*}
	&\brank(p_{A,2k,n-(i+1)}) \leq \brank(p_{A,2k,n-i}) + \brank(p_{A',2k,(n-1) - i}) \\
	&\leq 2^{i -1}(n + 2(k-1)D^{k-1}) + 2^{i-1}((n-1) + 2(k-1)D^{k-1}) \leq 2^{i}(n + 2(k-1)D^{k-1}).
	\end{align*}
	Therefore the statement (i) holds.
	
	(ii) We have
	\[
		p_{A,2k,n-2k}(x) = ({\det} (A(x) + \Lambda_n^{n-2k}))^{(2k)} = \det(A_{2k}(x)),
	\] 
	where $A_{2k}$ is the leading principal submatrix of $A$ of size $2k$.
	Given $I \subseteq [1,2k]$ with $|I| = k$,
	we denote by $B_I$ the square submatrix of $A_{2k}$ consisting of rows and columns corresponding to indices $I$ and $[1,k]$, respectively.
	Also, we denote by $B_{\bar{I}}(x)$ the square submatrix of $A_{2k}(x)$ 
	with row and column indices $[1,2k] \setminus I$ and $[k+1,2k]$, respectively. 
	Then the following is an elementary formula for determinant.
	\[
	{\det}(A_{2k}) = \sum_{I \subseteq  [1,2k] : |I|=k} \sign(I) {\det} (B_I )\cdot {\det}(B_{\bar{I}}),
	\]
	where $\sign(I) \in \{1, -1\}$.
	Since ${\det}(B_I ), {\det}(B_{\bar{I}}) \in K[x]^{(k)}$ for all $I\subseteq  [1,2k]$ with $|I|=k$, we have $\brank(p_{A,2k,r}) \leq \binom{2k}{k}$.
\end{proof}
\end{lem}
We give a proof of Theorem~\ref{thm_main_dc}.
\begin{proof}[Proof of Theorem~\ref{thm_main_dc}]
	Suppose that $\dc(p) = n$.
	By Lemma~\ref{lem_key1}, for all $x_0 \in {\rm Zeros}(p)$, there exist 
	a linear polynomial matrix $A \in (K[\xx])^{n \times n}$ and $r \in [0,n-1]$
	such that $p_{x_0}(x) = \det(A(x) + \Lambda^r_n)$.
	Then $p_{x_0}^{(2k)}(x) = (\det(A(x) + \Lambda^r_n))^{(2k)} = p_{A,2k,r}$.	
	If $r < n-2k$ or $n=1$, then $p_{A,2k,r} = 0$ and $\brank(p_{x_0}^{(2k)}) = 0$.
	The statement is trivial in this case,
	and we assume that $r \in [n-2k , n-1]$ and $n \geq 2$.
	By Lemma~\ref{lem_ABP_p} it holds that
	\[
		\brank(p_{A,2k,r}) \leq \max\{ 2^{2k -2}(n + 2(k-1)D^{k-1}), \binom{2k}{k}  \} 
		\leq 2^{2k - 2}(n + 2(k-1)D^{k-1}),
	\]
	since $k \leq D$, $n \geq 2$, and $\binom{2k}{k} \leq (2k)^k \leq  2^{2k-2}(n + 2(k-1)D^{k-1})$.
	Therefore it holds that 
	\[
	\frac{1 }{2^{2k-2}}\brank(p_{x_0}^{(2k)}) - 2(k-1) D^{k-1} \leq n = \dc(p).
	\]
\end{proof}

%% file: brank_proof2.tex
\subsection*{Proof of Proposition~\ref{prop_ABP_p_n-1}}
%
%To expand $\det(A(x) + \Lambda^r_n)$ with respect to $a_{i,j}(x)$, 
We denote $p_{A,k,n-1} := ({\det}( A(x) + \Lambda^{n-1}_n))^{(k)}$ by $p_{A,k}$ for notational simplicity.
Let us define the following notions about clows and cycles on a vertex set. The former is a terminology of Mahajan and Vinay~\cite{Mahajan1997Dca}.
\begin{itemize}
\item Let $V_n := [1,n]$, and we call elements of $V_n$ as {\em vertices}. 
\item A {\em clow} (standing for closed walk) on $V_n$ is an ordered tuple of vertices 
$\langle v_1 , v_2,\dots,v_l \rangle$ such that $v_1 < v_i$ for $i=2,\dots,l$.
The vertex $v_1$ is referred to as the {\em head} of $c$, and $l$ is called the {\em length} of $c$.
If all vertices $v_1,v_2,\dots,v_l$ are distinct, the clow is particularly called a {\em cycle}.
Note that $\langle v\rangle$ is a cycle for all $v \in V_n$.
\item Given a clow $c=\langle v_1 , v_2,\dots,v_l \rangle$ and a linear polynomial matrix $A(x) := (a_{i,j}(x)) \in (K[\xx])^{n \times n}$, $a_c \in K[\xx]^{(l)}$ is defined as $a_c := \prod_{i=1}^{l} a_{v_i, v_{i+1}}$, with identification $v_{l+1} = v_1$.
\item A {\em clow sequence} $C=(c_1,c_2,\dots,c_m)$ is an ordered tuple of clows $c_1,\dots,c_m$, where the head of $c_i$ is strictly less than the head of $c_j$ if $i<j$. The {\em size} $\ell(C)$ of $C$ is defined by $\ell(C):=m$.
\item We denote by $\Clow_{n,k}$ the set of clow sequences $C$ on $V_n$ such that the sum of length of all clows in $C$ is $k$. 
The subset $\cC_{n,k} \subseteq \Clow_{n,k}$ consists of vertex disjoint clow sequences $C'$ where each clow in $C'$ is a cycle. Since every element $C'$ of $\cC_{n,k}$ includes exactly $k$ distinct vertices, we call $C'$ as a {\em cycle $k$-cover}.
\item The {\em sign} of $C \in \Clow_{n,k}$ is defined as $(-1)^{n+ \ell(C)}$. 
\item We denote by $\Clow_{n,k,1} \subseteq \Clow_{n,k}$ the set of clow sequences which includes the vertex $1 \in V_n$. 
Also, $\cC_{n,k,1} \subseteq \Clow_{n,k,1}$ is defined as the set of cycle $k$-covers which include the vertices $1 \in V_n$. 
Given $C \in \Clow_{n,k,1}$, we denote by $c_1$ the unique clow in $C$ which includes the vertex $1$. 
\item Given a linear polynomial matrix $A(x) = (a_{i,j}(x)) \in (K[\xx])^{n \times n}$ and a clow sequence $C= (c_1,c_2,\dots,c_{\ell(C)}) \in \Clow_{n,k}$, 
we define a polynomial $a_C \in K[\xx]^{(k)}$ by $a_C := \prod_{i=1}^{\ell(C)} a_{c_i}$.
%\item Given a linear polynomial matrix $A \in (K[\xx])^{n \times n}$, $k \in \NN$, and $r \in [n-k,n-1]$, $p_{A,k,r}\in K[\xx]^{(k)}$ is defined as $p_{A,k,r} := \sum_{C \in \cC_{n,k,n-r}} \sign(C) a_C $.
\end{itemize}
\begin{lem}\label{lem_key2}
Given a linear polynomial matrix $A \in (K[\xx])^{n \times n}$, it holds that
\begin{align*}
p_{A,k}(x) = (-1)^{n-k} \sum_{C \in \cC_{n,k}} \sign(C) a_C
%{\det}( A(x) + \Lambda^r_n) = \sum_{k=n-r}^n (-1)^{n-k}
% p_{A,k,r}(x).
\end{align*}
\begin{proof}
We expand ${\det}( A(x) + \Lambda^{n-1}_n) = \sum_{i=1}^{n} p_{A,i}(x) $ with respect to $a_{i,j}(x)$.
Given $U$ with $\{1 \} \subseteq U \subseteq V_n$, $A_{U}$ is defined as the principal submatrix of $A$ consisting of rows and columns having indices in $U$.
Since determinant is a multilinear function and $\Lambda^{n-1}_n$ has nonzero entries which is equal to $1$ only on diagonal, it follows that
\begin{align}
{\det}( A(x) + \Lambda^r_n) 
= \sum_{U: \{1 \} \subseteq U \subseteq V_n} {\det } (A_U(x)) = \sum_{k=1}^{n}\sum_{\substack{ U:|U|= k \\ \{1\} \subseteq U \subseteq V_n }} {\det } (A_U(x)) . \label{eq_lem2_1}
\end{align}
For $U \subseteq V_n$, denote by $\cC_{n,U} \subseteq \cC_{n,|U|}$ the set of cycle $|U|$-covers which include all vertices in $U$.
Since elements in $\cC_{n,U}$ consist of vertices in $U$, we can regard them as cycle $|U|$-covers on $U$.
Let $\mathfrak{S}_U$ be the set of permutations on the finite set $U$.
Observe that there is a natural one-to-one correspondence between permutations $\sigma \in \mathfrak{S}_U$ and cycle $|U|$-covers $C \in \cC_{n,U}$, satisfying 
$\sign(\sigma) = (-1)^{n-|U|}\sign(C)$ and
$\prod_{i\in U} a_{i,\sigma(i)} = a_C$.
Therefore from the definition of determinant it follows that 
\begin{align}
{\det } (A_U) = \sum_{\sigma \in \mathfrak{S}_{U}} \sign(\sigma) \prod_{i\in U} a_{i,\sigma(i)} = (-1)^{n-|U|} \sum_{C \in \cC_{n,U}}\sign(C) a_C.
\label{eq_lem2_2}
\end{align}
By (\ref{eq_lem2_1}), (\ref{eq_lem2_2}), we obtain
\begin{align*}
{\det}( A(x) + \Lambda^r_n) 
&= \sum_{k=1}^{n} (-1)^{n-k} \sum_{\substack{U: |U|=k \\ \{1\} \subseteq U \subseteq V_n }}\sum_{C \in \cC_{n,U}}\sign(C) a_C(x) \\
&= \sum_{k=1}^{n} (-1)^{n-k} \sum_{C \in \cC_{n,k, n-r}}\sign(C) a_C(x) = \sum_{k=1}^{n} (-1)^{n-k} p_{A,k,r}(x).
\end{align*}
The second equation holds since $\cC_{n,k,n-r}$ is the disjoint union of $\cC_{n,U}$ over $U$ satisfying $\{1 \} \subseteq U \subseteq V_n$ and $|U|=k$.
Therefore it holds that $p_{A,k} = ({\det}( A(x) + \Lambda^r_n))^{(k)} = (-1)^{n-k} p_{A,k,r}(x)$.
\end{proof}
\end{lem}
The following lemma was essentially given in~\cite[Section 3]{Valiant1992WiB}, and proved in full detail in~\cite{Mahajan1997Dca, Mahajan1999Doa}.
%The proof is included in Appendix, which is almost the same as that in~\cite[Lemma 1]{Rote2001Dfa}.
%
%
%
\begin{lem}\label{lem_key3}
It holds that
\begin{align*}
p_{A,k}
= (-1)^{n-k} \sum_{C \in \Clow_{n,k,1}} \sign(C) a_C .
\end{align*}
%The proof of Lemma~\ref{lem_key3} is the following.
\begin{proof}%[Proof of Lemma~\ref{lem_key3}]
	Since $\cC_{n,k,1} \subseteq \Clow_{n,k,1}$, by the definition of $p_{A,k,1}$ it is enough to show that 
	\begin{align}
	\sum_{C \in \Clow_{n,k,1} \setminus \cC_{n,k,1}} \sign(C) a_C (x)=0.\label{eq_key3}
	\end{align} 
	Observe that $C \in \Clow_{n,k,1}$ is not in $\cC_{n,k,1}$ if and only if $C$ has 
	a repetition of vertices. 
	Suppose that $C=( c_1,c_2,\dots,c_{m} ) \in \Clow_{n,k,1}$ have a repetition of vertices. 
	Let $l \leq m$ be the maximum number such that $( c_{l},\dots,c_{m} )$ has a repetition but $(c_{l+1},\dots,c_{m})$ does not.
	Let $c_l = \langle v_1,v_2,\dots,v_i\rangle$ and $v_j$ be the first element in $c_l$ which is either 
	(1) equal to one of $v_{t}$ where $1 < t < j$, or
	(2) equal to an element in one of $c_{l+1},\dots,c_{m}$.
	Precisely one of them will occur.
	
	In the case (1), let $\langle u_t,u_{t+1},\dots,u_{j-1} \rangle$ be the clow obtained by cyclically reordering the vertex sequence $v_t,v_{t+1}\dots,v_{j-1}$ (these vertices are all distinct and there is a unique minimum element).
	We replace $c_l$ by two clows $\langle v_1,\dots,v_{t-1}, v_j ,\dots,v_i \rangle$ and $\langle u_t,\dots,u_{j-1} \rangle$, which are made by separation around $v_t$.
	
	In the case (2), let $c'= \langle u_1,\dots,u_{s-1}, v_{j}, u_{s+1}, \dots, u_{i'} \rangle$ be the clow including $v_j$.
	We replace $c_l$ and $c'$ by the single clow $\langle v_1,\dots,v_j, u_{s+1},\dots, u_{i'}, u_1 ,\dots,u_{s-1}, v_j , v_{j+1},\dots,v_i \rangle $, which is constructed by the insertion of $c'$ into $c$ around the common vertex $v_t$.
	It can be verified that the procedures (1) and (2) are inverses of each other, which result in a one-to-one correspondence in clow sequences with repetition.
	If $C \in \Clow_{n,k,1} \setminus \cC_{n,k,1}$ is converted to $C'\in \Clow_{n,k,1}\setminus \cC_{n,k,1}$ by the above procedure, 
	then $a_C(x) = a_{C'}(x)$ and $\sign(C) = -\sign(C')$. Thus the equation (\ref{eq_key3}) holds.
\end{proof}
\end{lem}
The number of polynomials needed to span $K[\xx]^{(k)}$ as a $K$-vector space is $\binom{D+k-1}{k}$, and we define $s_k := \binom{D+k-1}{k}$.
The following is a general property of polynomials.
\begin{lem}\label{lem_basic}
	For $p \in K[\xx]^{(k)}$ and $m \leq k$,
	there exist
	$2s_{m}$ polynomials $f_1,f_2,\dots,f_{s_{m}} \in K[\xx]^{(m)}$ and $g_1,g_2,\dots,g_{s_{m}} \in K[\xx]^{(k-m)}$ such that $p = \sum_{i=1}^{s_{m}} f_i g_i$.
	\begin{proof}
		Let $P_{m} := \{x_1^{j_1}x_2^{j_2}\cdots x_D^{j_D} \mid j_1,j_2,\dots,j_D \in \ZZ_+,
		\ \sum_{l=1}^{D} j_l = m \}$ be the set of monomials with degree $m$. Take distinct $f_i \in P_{m}$ for $i=1,2,\dots,s_{m}$. Since each term in $p$ is divisible by at least one of $f_i$, we can obtain a decomposition $p = \sum_{i=1}^{s_m} f_i g_i$.
	\end{proof}
\end{lem}

For $t=1,2,\dots,2k$, let $q_{A,2k,t} \in K[\xx]^{(2k)}$ be defined by
\[
q_{A,2k,t} := \sum_{\substack{C \in \Clow_{n,2k,1} \\ |c_1| = t}} \sign(C) a_C ,
\]
where $c_1$ is the unique clow in $C$ including the vertex $1$, and $|c_1|$ is the length of $c_1$.
Then it holds that 
$p_{A,2k} = (-1)^{n-2k} \sum_{t=1}^{2k} q_{A,2k,t}$.
Let $\sF_t$ be the set of clows with length $t$ which include the vertex $1 \in V_n$.

\begin{lem}\label{lem_q_decomposition}
	(i) For $t=2k$, it holds that $\brank(q_{A,2k,2k}) \leq n-1$.
	
	(ii) For $t=1,2,\dots,2k-1$, let $t' := \min\{t,2k-t \}$.
	Then it holds that $\brank(q_{A,2k,t}) \leq s_{k-t'}$.
\begin{proof}
	(i) We have $q_{A,2k,2k} = (-1)^{n+1}\sum_{c \in \sF_{2k}} a_c$.
	For $v =2,3,\dots,n$, let $\sF_{2k,v} \subseteq \sF_{2k}$ be the set of clows whose $(k+1)$th vertices are equal to $v$. 
	Note that the $(k+1)$th vertex of a clow $c \in \sF$ is one of $\{2,3,\dots,n\}$,
	since the head of $c$ is $1$. 
	Then we have $q_{A,2k,2k} = (-1)^{n+1}\sum_{v \in [2,n]} \sum_{c \in \sF_{2k,v}} a_c$.
	Let $\mathcal{R}_{v}$ and $\mathcal{R}_{v}'$ be the sets of $k +1$ vertex sequences $R = (1,u_2,u_3,\dots,  u_{k},v)$ and $R' = (v,u_2',u_3'\dots,  u_{k}', 1)$, respectively, such that $u_2,u_3,\dots, u_{m}$, $u_2',u_3', \dots, u_{k+1}' \in \{2,3,\dots,n \}$. 
	We define $a_R := \prod_{i=1}^k a_{u_i,u_{i+1}}$ for $R = (u_1,u_2,\dots,  u_{k+1})$ in $\mathcal{R}_{v}$ or $\mathcal{R}_{v}'$.
	%For $R = (u_1,u_2,\dots,  u_{k+1})$ in $\mathcal{R}_{v}$ or $\mathcal{R}_{v}'$, we define $a_R := \prod_{i=1}^k a_{u_i,u_{i+1}}$.
	Then there is a one-to-one correspondence between $\sF_{2k,v}$ and $\mathcal{R}_{v} \times \mathcal{R}_{v}'$ by the following correspondence
	\begin{align*}
	&\sF_{2k,v} \ni \langle 1, u_2,u_3,\dots,  u_{k},v, u_{k+2},u_{k+3},\dots,  u_{2k} \rangle \\ 
	&\mapsto ((1, u_2,u_3,\dots,  u_{k},v), (v, u_{k+2},u_{k+3},\dots,  u_{2k},1)) \in \mathcal{R}_{v} \times \mathcal{R}_{v}'.
	\end{align*}
	Therefore it holds that
	\[
	\sum_{c \in \sF_{k,v}} a_c
	=\left(\sum_{R \in \mathcal{R}_{v}} a_R \right) 
	\left( \sum_{R' \in \mathcal{R}_{v}'} a_{R'}\right).
	\]
	Therefore it holds that $\brank(q_{A,2k,2k}) \leq \sum_{v \in [2,n]} \brank(\sum_{c \in \sF_{k,v}} a_c) \leq n-1$.

	(ii) Let $\sG_{2k-t} \subseteq \Clow_{n,2k-t}$ be the set of clow sequences which do not include the vertex $1 \in V_n$.
	Observe that a clow sequence $C \in \Clow_{n,2k,1}$ with $|c_1| = t$ is uniquely determined by a pair of a clow $c_1 \in \sF_t$ and a clow sequence $C' \in \sG_{k-t}$. Furthermore, both can be chosen independently. Therefore we have
	\[
	q_{A,2k,t} = \left(\sum_{c \in \sF_t} a_c\right)  \left(-\sum_{C' \in \sG_{2k-t}} \sign(C')a_{C'} \right).
	\]
	One of polynomial $(\sum_{c \in \sF_t} a_c)$ and $(-\sum_{C' \in \sG_{2k-t}} \sign(C')a_{C'})$ has degree $t'$, and we denote it by $\alpha_1$. The other is denoted by $\alpha_2$.
	%We denote by $q_1$ one of polynomials $(\sum_{c \in \sF_t} a_c)$ and $(-\sum_{C' \in \sG_{2k-t}} \sign(C')a_{C'})$
	%with degree $t'$, and by $q_2$ the other.
	If $t' = k$, we already have the decomposition
	$q_{A,2k,t} = \alpha_1 \alpha_2$, and $\brank(q_{A,2k,t})=1 \leq s_{k-t'}$.
	Otherwise, by Lemma~\ref{lem_basic} there exist $s_{k -t'}$ polynomials
	$\beta_1,\beta_2,\dots, \beta_{s_{k -t'}} \in K[\xx]^{( k -t')}$ and $\gamma_1,\gamma_2,\dots, \gamma_{s_{k -t'}} \in K[\xx]^{( k)}$ such that 
	$\alpha_2 = \sum_{i=1}^{s_{ k -t'}} \beta_i \gamma_i$.
	Then we have
	\[
	q_{A,2k,t} = \sum_{i=1}^{s_{k -t'}} (\alpha_1  \beta_i) \cdot \gamma_i,
	\]
	and $\brank(q_{A,2k,t}) \leq s_{k-t'}$.
\end{proof}	
\end{lem}
\begin{proof}[Proof of Proposition~\ref{prop_ABP_p_n-1}]
	By definition, $p_{A,2k} = (-1)^{n-2k} \sum_{t=1}^{2k} q_{A,2k,t}$.
	By Lemma~\ref{lem_q_decomposition}, it holds that
	\[
	\brank(p_{A,2k}) \leq\sum_{t=1}^{2k} \brank (q_{A,2k,t}) 
	\leq (n-1) + 2\sum_{t=1}^{k-1}s_t \leq n + 2(k-1).
	\]
\end{proof}

%% file: brank_acknowledgement.tex
\section*{Acknowledgements}
We are deeply grateful to Hiroshi Hirai, Kyo Nishiyama, Jun Tarui and Takeshi Tokuyama
for helpful comments improving the presentation of this paper.
We are indebted to Susumu Ariki for suggesting a formulation of the bi-polynomial rank,
whereas our original formulation was based on tensors of higher-order differentials,
and was quite complicated.
We appreciate Hiroshi Hirai for pointing out the relation 
between our original approach and the theory of the concave minimization.
This work evolved from discussions of monthly GCT seminars.
We thank all the members (S. Ariki, N. Enomoto, H. Hirai, H. Matsumoto, K. Nishiyama, J. Tarui and T. Tokuyama) of the seminar.
%I would like to express the deepest appreciation to 
%the author's supervisor H. Hirai.
%Without his guidance, this paper would not have materialized.
%This work was inspired by discussions in the monthly GCT seminar organized
%by S. Ariki, N. Enomoto, H. Hirai, H. Matsumoto, K. Nishiyama, J. Tarui and
%T. Tokuyama. The author is also grateful to the seminar members for their valuable suggestions on preliminary version of the draft.
The author is supported by the ELC project (Grant-in-Aid for
Scientific Research on Innovative Areas MEXT Japan),
and is partially supported by KAKENHI(26330023).